\documentclass[draftcls,onecolumn,12pt]{IEEEtran}

\usepackage [latin1]{inputenc}
\usepackage[usenames,dvipsnames]{xcolor}
\usepackage[tbtags]{amsmath}
\usepackage{amsthm}
\usepackage{subfigure}
\usepackage{graphicx}
\usepackage{amssymb}
\usepackage{caption}
\usepackage{threeparttable}

\newtheorem{theorem}{Theorem}

\newtheorem{lemma}{Lemma}

\newtheorem{remark}{Remark}

\usepackage[numbers,sort&compress]{natbib}

\begin{document}

\title{A Tighter Upper Bound of the Expansion Factor for Universal Coding of Integers and Its Code Constructions}

\author{Wei Yan,
        Sian-Jheng Lin,~\IEEEmembership{Member,~IEEE}
        \thanks{Yan and Lin are with the School of Information Science and Technology, University of Science and Technology of China (USTC), China. (e-mail: yan1993@mail.ustc.edu.cn; sjlin@ustc.edu.cn).}
}

\maketitle
\begin{abstract}
In entropy coding, universal coding of integers~(UCI) is a binary universal prefix code, such that the ratio of the expected codeword length to $\max\{1, H(P)\}$ is less than or equal to a constant expansion factor $K_{\mathcal{C}}$ for any probability distribution $P$, where $H(P)$ is the Shannon entropy of $P$. $K_{\mathcal{C}}^{*}$ is the infimum of the set of expansion factors.
The optimal UCI is defined as a class of UCI possessing the smallest $K_{\mathcal{C}}^{*}$. Based on prior research, the range of $K_{\mathcal{C}}^{*}$ for the optimal UCI is $2\leq K_{\mathcal{C}}^{*}\leq 2.75$. Currently, the code constructions achieve $K_{\mathcal{C}}=2.75$ for UCI and $K_{\mathcal{C}}=3.5$ for asymptotically optimal UCI.
In this paper, we propose a class of UCI, termed $\iota$ code, to achieve $K_{\mathcal{C}}=2.5$. This further narrows the range of $K_{\mathcal{C}}^{*}$ to $2\leq K_{\mathcal{C}}^{*}\leq 2.5$.
Next, a family of asymptotically optimal UCIs is presented, where their expansion factor infinitely approaches $2.5$.
Finally, a more precise range of $K_{\mathcal{C}}^{*}$ for the classic UCIs is discussed.
\end{abstract}

\section{Introduction}
In entropy coding, when the probability distribution of sources is unknown and difficult to measure, some entropy coding, such as arithmetic coding~\cite{AC79,AC84} and Huffman coding~\cite{H52}, cannot be applied to compress the source. In this case, universal source coding~\cite{USC} is a common way to encode the data, and LZ series algorithms~\cite{LZ77,LZ78,LZW} is one of the well-known algorithms of universal source coding.
However, there is no universal source coding for infinite alphabet and discrete memoryless sources~\cite{No94}.
Universal coding of integers~(UCI) is a universal code for infinite alphabet and discrete memoryless sources.
UCIs have been applied in widespread applications, such as unbounded search problems~\cite{BY76,AHK97}, inverted file index~\cite{06}, inductive inference~\cite{DCC21} and biological sequencing data compression~\cite{DNA10,DNA13}.

Prefix coding is a class of variable-length code that no codeword is a prefix of any other codeword.
Binary coding means that the coding alphabet is $\{0,1\}$.
Elias~\cite{Elias75} defined UCI as a binary universal prefix code, such that the ratio of the expected codeword length to $\max\{1,H(P)\}$ is less than or equal to a constant expansion factor $K_{\mathcal{C}}$ for any probability distribution $P$, where $H(P)$ is the Shannon entropy of $P$.
Many UCIs have been proposed and most of they can be divided into the following two categories~\cite{C1990,AS17} (For example, \emph{group strategy}~\cite{A1993} is the exception).
\begin{enumerate}
\item \emph{message length strategy:}  This strategy is to encode a positive integer $n$ into two parts.
The suffix part of length $L$ represents $n$, and the prefix part standing for $L$ (The prefix part can be further subdivided).
The coding of this strategy was proposed in~\cite{Elias75,L68,ER78,S80,Y00,yan}.
\item \emph{flag strategy:}  This strategy is to select a special sequence, called flag, to determine the end of a codeword. The flag is not allowed to appear within a codeword.
The coding of this strategy was proposed in~\cite{L81,AF87,W88,YO91,AS17}.
\end{enumerate}

Recently, Yan and Lin~\cite{yan} first studied the range of $K_{\mathcal{C}}$. First, the authors defined \emph{optimal UCI}, which is a class of UCI with the smallest $K_{\mathcal{C}}^{*}\triangleq\inf\{K_{\mathcal{C}}\}$. It is showed that the optimal UCI is in the range $2\leq K_{\mathcal{C}}^{*}\leq 2.75$, where $K_{\mathcal{C}}^{*}= 2.75$ is achieved by $\eta$ code~\cite{yan}. In particular, for the asymptotically optimal UCI, the smallest expansion factor is $K_{\mathcal{C}}=3.5$, which is achieved by $\theta$ code~\cite{yan} and Elias $\omega$ code~\cite{Elias75}.

In this paper, we further narrow the range of $K_{\mathcal{C}}^{*}$ of the optimal UCI. The contributions of this paper are listed below.
\begin{enumerate}
\item A class of UCI, but not asymptotically optimal, with $K_{\mathcal{C}}=2.5$ is presented. This reduces the upper bound of $K_{\mathcal{C}}^{*}$ from $2.75$ to $2.5$.
\item A family of asymptotically optimal UCIs is proposed, where $K_{\mathcal{C}}$ infinitely approaches $2.5$.
\item The range of $K_{\mathcal{C}}^{*}$ for some classic UCIs is discussed (see Table \ref{tab3}).
\end{enumerate}

In the rest of this paper, Section \ref{pre} introduces some background knowledge.
Section \ref{sec_la} presents the main theorem of this paper.
Section \ref{sec_lb} proposes a class of UCI to achieve $K_{\mathcal{C}}=2.5$.
A family of asymptotically optimal UCIs is proposed in Section \ref{sec_new1}.
Section \ref{sec_new} gives a more precise range of $K_{\mathcal{C}}^{*}$ for the classic UCIs.
Section \ref{sec_con} concludes this work.

\section{Preliminaries}\label{pre}
\subsection{The definitions of UCI and asymptotically optimal UCI}
Elias~\cite{Elias75} treated the coding problem as follows.
Let $\mathcal{C}$ be a given binary prefix coding of the positive integers $\mathcal{N}\triangleq\{1,2,\cdots,m,\cdots\}$.
Let $L_{\mathcal{C}}(\cdot)$ denote the length function of $\mathcal{C}$ (i.e., $L_{\mathcal{C}}(m)=|\mathcal{C}(m)|$, for all $m\in\mathcal{N}$).
Let $P$ denote any probability distribution of $\mathcal{N}$ (i.e., $\sum_{n=1}^{\infty}P(n)=1$, and $P(m)\geq 0$, for all $m\in \mathcal{N}$).
In UCI, the source meets the probability distribution
\begin{equation}\label{eq1}
P(m)\geq P(m+1),
\end{equation}
for all $m\in \mathcal{N}$. Let $E_{P}(L_{\mathcal{C}})=\sum_{n=1}^{\infty}L_{\mathcal{C}}(n)P(n)$ be the expected codeword length for $\mathcal{C}$,
and let $H(P)=-\sum_{n=1}^{\infty}P(n)\log_{2}P(n)$ denote the entropy of $P$.
Elias~\cite{Elias75} defined $\mathcal{C}$ to be \emph{universal} if there is a constant $K_{\mathcal{C}}$ such that
\begin{equation}\label{eq2}
\frac{E_{P}(L_{\mathcal{C}})}{\max\{1,H(P)\}}\leq K_{\mathcal{C}},
\end{equation}
for all $P$ with finite entropy, where $K_{\mathcal{C}}$ is the expansion factor.
Furthermore, $\mathcal{C}$ is called \emph{asymptotically optimal} if $\mathcal{C}$ is universal and a function $R_{\mathcal{C}}(\cdot)$ exists such that
\begin{equation}
\lim\limits_{H(P)\to +\infty}R_{\mathcal{C}}(H(P))=1,
\end{equation}
and
\begin{equation}
\frac{E_{P}(L_{\mathcal{C}})}{\max\{1,H(P)\}}\leq R_{\mathcal{C}}(H(P)),
\end{equation}
for all $P$ with finite entropy.
\subsection{Some classic UCIs}
In this subsection, we briefly introduce five classic UCIs, termed $\gamma$ code, $\delta$ code, $\omega$ code, $\eta$ code, and $\theta$ code. For the specific structure of classic UCIs, please refer to~\cite{Elias75,yan}. First, the codeword lengths and the range of $K_{\mathcal{C}}^{*}$ of these UCIs are listed in Table \ref{tab5}.
And the five classic UCIs all satisfy $L_{\mathcal{C}}(1)=1$.
Next, the following theorem can be used to judge whether a UCI is asymptotically optimal.
\begin{table}[t]
\begin{threeparttable}
\centering
\caption{The codeword lengths and ranges of $K_{\mathcal{C}}^{*}$ of some classic UCIs}\label{tab5}
\begin{tabular}{|c|c|c|c|}
\hline
Code  & The codeword lengths for $2\leq m\in \mathcal{N}$ & The range of $K_{\mathcal{C}}^{*}$ &  Asymptotically optimal   \\
\hline
$\gamma$ code  &  $L_{\gamma}(m)=1 + 2\lfloor \log_{2}m\rfloor$                                                       & $K_{\gamma}^{*}=3$         &  No   \\
$\delta$ code  &  $L_{\delta}(m)=1 + \lfloor \log_{2}m\rfloor + 2\lfloor \log_{2}(1+\lfloor \log_{2}m\rfloor)\rfloor$ & $2.5 \leq K_{\delta}^{*}\leq 4$  &  Yes  \\
$\omega$ code  &  $L_{\omega}(m)=1+\sum_{n=1}^{s}(\lambda^{n}(m)+1)$ \ \tnote{1}                                        &  $2.1 < K_{\omega}^{*}\leq 3.5 $     &  Yes  \\
$\eta$ code    &  $L_{\eta}(m)=3+\lfloor\log_2(m-1)\rfloor+\lfloor \frac{\lfloor\log_2(m-1)\rfloor}{2} \rfloor$       & $2.5 \leq K_{\eta}^{*}\leq 2.75$ &  No   \\
$\theta$ code  &  $L_{\theta}(m)=3+\lfloor \log_{2}m\rfloor+\lfloor\log_{2}\lfloor \log_{2}m\rfloor \rfloor+\lfloor\frac{\lfloor\log_{2}\lfloor \log_{2}m\rfloor \rfloor}{2}\rfloor$ & $ 2.5 \leq K_{\theta}^{*}\leq 3.5$ &  Yes  \\
\hline
\end{tabular}
\begin{tablenotes}
\footnotesize
\item[1] $\lambda(m)\triangleq \lfloor \log_{2}m\rfloor$, $\lambda^{n}$ is the $n$-fold compositions of function $\lambda$, and $s=s(m)\in \mathcal{N}$ is a uniquely integer satisfying $\lambda^{s}(m)=1$.
\end{tablenotes}
\end{threeparttable}
\end{table}
\begin{theorem}~\cite{Elias75,L81}\label{thm3}
Given a UCI $\mathcal{C}$, the function $L_{\mathcal{C}}(\cdot)$ satisfies $L_{\mathcal{C}}(m)\geq c+b\lfloor \log_{2}m \rfloor$ for all $m\in \mathcal{N}$,
where $b$ is a constant greater than $1$ and $c$ is a constant.
Then, $\mathcal{C}$ is not asymptotically optimal.
\end{theorem}

\section{The main theorem}\label{sec_la}
In this section, we present the main theorem of this paper.
First, a related lemma is provided, then the theorem is given.
\begin{lemma}\label{lemma1}
Given an any probability distribution $P=(P(1),P(2),\cdots,P(m),\cdots)$, then
\begin{itemize}
\item[(1)]$H(P)\geq-\log_{2}P(1)$;
\item[(2)]If $H(P)<1$, then $P(1)>\frac{1}{2}$.
\end{itemize}
\end{lemma}
\begin{proof}
\begin{itemize}
\item[(1)]\begin{equation}
\begin{aligned}
H(P)&=\sum_{n=1}^{\infty}P(n)\log_{2}\frac{1}{P(n)}     \\
    &\geq \sum_{n=1}^{\infty}P(n)\log_{2}\frac{1}{P(1)}  \\
    &=-\log_{2}P(1).   \\
\end{aligned}
\end{equation}
\item[(2)] When $H(P)<1$, then
\begin{equation}
-\log_{2}P(1)\leq H(P)<1\Rightarrow P(1)>\frac{1}{2}.
\end{equation}
\end{itemize}
\end{proof}
\begin{theorem}\label{thm4}
Given a prefix code $\mathcal{C}$, the function $L_{\mathcal{C}}(\cdot)$ satisfies $L_{\mathcal{C}}(1)=1$ and $L_{\mathcal{C}}(m)\leq b+1+b\lfloor\log_{2}m\rfloor$ for all $2\leq m\in \mathcal{N}$,
where the constant $b$ is in the range $1\leq b\leq \frac{9}{4}$.
Then,
\begin{equation}
\frac{E_{P}(L_{\mathcal{C}})}{\max\{1,H(P)\}}\leq b+1;
\end{equation}
that is, $\mathcal{C}$ is a UCI and $K_{\mathcal{C}}^{*}\leq b+1$.
\end{theorem}
\begin{proof}
Due to
\begin{equation}
mP(m)\leq\sum_{n=1}^{m}P(n)\leq \sum_{n=1}^{\infty}P(n)=1,
\end{equation}
we have $m\leq\frac{1}{P(m)}$ for all $m\in \mathcal{N}$.
Thus, we obtain
\begin{equation}\label{eq9}
\begin{aligned}
\sum_{n=2}^{\infty}P(n)\log_{2}n &\leq \sum_{n=2}^{\infty}P(n)\log_{2}\frac{1}{P(n)}    \\
& = H(P)+P(1)\log_{2}P(1).
\end{aligned}
\end{equation}
The expected codeword length is
\begin{equation}
\begin{aligned}
  E_{P}(L_{\mathcal{C}})&\leq P(1)+\sum_{n=2}^{\infty}P(n)(b+1+b\lfloor\log_{2}n\rfloor) \\
             &= b+1-bP(1)+b\sum_{n=2}^{\infty}P(n)\lfloor\log_{2}n\rfloor \\
             &\leq b+1-bP(1)+b\sum_{n=2}^{\infty}P(n)\log_{2}n\\
             &\leq b+1-bP(1)+bH(P)+bP(1)\log_{2}P(1). \\
\end{aligned}
\end{equation}
We consider three cases below.
\begin{enumerate}
\item Case $H(P)<1$: In this case, we obtain $P(1)>\frac{1}{2}$ from Lemma \ref{lemma1}. Further, we have
\begin{equation}
\begin{aligned}
\frac{E_{P}(L_{\mathcal{C}})}{\max\{1,H(P)\}}&\leq b+1-bP(1)+bH(P)+bP(1)\log_{2}P(1)   \\
                                             &\leq 2b+1-bP(1)+bP(1)\log_{2}P(1).     \\
\end{aligned}
\end{equation}
Let $g_{1}(x)\triangleq 2b+1-bx+bx\log_{2}x$. We only need to prove that $g_{1}(x)\leq b+1$ over interval $[\frac{1}{2},1]$.
We know that the curve of $g_{1}$ is U-shaped over interval $[\frac{1}{2},1]$ by its derivative.
Thus, we have $g_{1}(x)\leq\max\{g_{1}(\frac{1}{2}),g_{1}(1)\}=b+1$ over interval $[\frac{1}{2},1]$.
\item Case $H(P)\geq1$ and $P(1)\geq0.5$: In this case, we have
\begin{equation}
\begin{aligned}
\frac{E_{P}(L_{\mathcal{C}})}{\max\{1,H(P)\}}&\leq\frac{b+1-bP(1)+bH(P)+bP(1)\log_{2}P(1)}{H(P)}    \\
                                             & = b+\frac{b+1-bP(1)+bP(1)\log_{2}P(1)}{H(P)}    \\
                                             & \leq 2b+1-bP(1)+bP(1)\log_{2}P(1)     \\
                                             & \overset{(a)}{\leq} b+1.     \\
\end{aligned}
\end{equation}
where $(a)$ is due to $g_{1}(x)\leq b+1$ over interval $[\frac{1}{2},1]$.
\item Case $H(P)\geq1$ and $P(1)<0.5$: In this case, we obtain
\begin{equation}
\begin{aligned}
\frac{E_{P}(L_{\mathcal{C}})}{\max\{1,H(P)\}}&\leq b+\frac{b+1-bP(1)+bP(1)\log_{2}P(1)}{H(P)}    \\
                                             &\overset{(a)}{\leq} b+\frac{b+1-bP(1)+bP(1)\log_{2}P(1)}{-\log_{2}P(1)}.    \\
\end{aligned}
\end{equation}
where $(a)$ is due to Lemma \ref{lemma1}. Let
\begin{equation}
\begin{aligned}
g_{2}(x)&\triangleq \frac{b+1-bx+bx\log_{2}x}{-\log_{2}x}    \\
        & =\ln{2} \cdot \frac{bx-b-1}{\ln{x}}-bx.    \\
\end{aligned}
\end{equation}
We need to prove that $g_{2}(x)\leq 1$ over interval $(0,\frac{1}{2})$.
We first prove that $g_{2}'(x)>0$ over interval $(0,\frac{1}{2})$.
Due to
\begin{equation}
g_{2}'(x)=\ln 2\cdot\frac{b\ln x-b+\frac{b+1}{x}}{(\ln x)^{2}}-b,
\end{equation}
then $g_{2}'(x)>0$ over interval $(0,\frac{1}{2})$ is equivalent to
$f(x)>\frac{1}{\ln 2}$ over interval $(0,\frac{1}{2})$, where
\begin{equation}
f(x)\triangleq\frac{\ln x-1+\frac{b+1}{bx}}{(\ln x)^{2}}.
\end{equation}
Finally, we obtain
\begin{equation}
\begin{aligned}
f'(x)&=\frac{-\ln x}{x^{2}(\ln x)^{4}}h(x)    \\
     &\triangleq \frac{-\ln x}{x^{2}(\ln x)^{4}}\left(x\ln x+\frac{b+1}{b}\ln x-2x+\frac{2b+2}{b}\right),   \\
h'(x)&=\ln x +\frac{b+1}{bx}-1,       \\
h''(x)&=\frac{1}{x^{2}}\left(x-\frac{b+1}{b}\right).  \\
\end{aligned}
\end{equation}
Due to $h''(x)<0$ over interval $(0,\frac{1}{2})$, we have
\begin{equation}
\begin{aligned}
h'(x)&>h'(\frac{1}{2})   \\
     &=\ln \frac{1}{2} +\frac{2b+2}{b}-1   \\
     &> -\ln 2-1+2   \\
     &>0,      \\
\end{aligned}
\end{equation}
over interval $(0,\frac{1}{2})$. Thus, $h(x)$ strictly increases over interval $(0,\frac{1}{2})$.
Due to
\begin{equation}
\begin{aligned}
h(0.19)&=0.19\ln 0.19+\frac{b+1}{b}(2+\ln 0.19)-2\times0.19  \\
       &\leq0.19\ln 0.19+2\times(2+\ln 0.19)-0.38   \\
       &<0       \\
\end{aligned}
\end{equation}
and
\begin{equation}
\begin{aligned}
h(0.24)&=0.24\ln 0.24+\frac{b+1}{b}(2+\ln 0.24)-2\times0.24  \\
       &\geq0.24\ln 0.24+\frac{13}{9}\times(2+\ln 0.24)-0.48  \\
       &>0,       \\
\end{aligned}
\end{equation}
there exists $x_{0}\in(0.19,0.24)$ such that $h(x_{0})=0$. Further, we have
$h(x)<0$ and $f'(x)<0$ over interval $(0,x_{0})$, $h(x)>0$ and $f'(x)>0$ over interval $(x_{0},\frac{1}{2})$.
And hence, $f(x)$ strictly decreases over interval $(0,x_{0})$ and $f(x)$ strictly increases over interval $(x_{0},\frac{1}{2})$.
Thus, we obtain
\begin{equation}
\begin{aligned}
f(x)&\geq f(x_{0})    \\
    &=\frac{1}{\ln x_{0}}-\frac{1}{(\ln x_{0})^{2}}+\frac{b+1}{b}\cdot\frac{1}{x_{0}(\ln x_{0})^{2}}  \\
    &>\frac{1}{\ln 0.24}-\frac{1}{(\ln 0.24)^{2}}+\frac{13}{9}\times\frac{1}{0.19(\ln 0.19)^{2}}   \\
    &>\frac{1}{\ln 2}      \\
\end{aligned}
\end{equation}
over interval $(0,\frac{1}{2})$.
Since $g_{2}'(x)>0$ over interval $(0,\frac{1}{2})$, we obtain
\begin{equation}
\begin{aligned}
g_{2}(x)&<g_{2}(\frac{1}{2})    \\
        & =\ln{2} \cdot \frac{\frac{b}{2}-b-1}{-\ln 2}-\frac{b}{2}   \\
        & =1,  \\
\end{aligned}
\end{equation}
for all $x\in(0,\frac{1}{2})$.
\end{enumerate}
The proof is completed.
\end{proof}
\begin{remark}
In Theorem \ref{thm4}, the feasible range $1\leq b\leq \frac{9}{4}$ is not tight.
The upper bound of $b$ is taken to be $\frac{9}{4}$ for the convenience of proving that
$f(x)>\frac{1}{\ln 2}$ over interval $(0,\frac{1}{2})$.
\end{remark}

When $b=1$ in Theorem \ref{thm4}, the theoretical lower bound of $K_{\mathcal{C}}^{*}$ of the optimal UCI in~\cite{yan} can be obtained.
In fact, there is no such prefix code when $1\leq b<1.5$.
\begin{theorem} \label{thm5}
There is no prefix code $\mathcal{C}$ such that $L_{\mathcal{C}}(1)=1$ and $L_{\mathcal{C}}(m)\leq b+1+b\lfloor\log_{2}m\rfloor$ for all $2\leq m\in \mathcal{N}$,
where $b$ is a constant less than $\frac{3}{2}$.
\end{theorem}
\begin{proof}
Suppose there is a prefix code $\mathcal{C}$ to meet the requirement.
\begin{enumerate}
\item For $m=2,3$, $L_{\mathcal{C}}(m)\leq 2b+1 <4$. Thus, $L_{\mathcal{C}}(m)\leq3$.
\item For $m=4,5,6,7$, $L_{\mathcal{C}}(m)\leq 3b+1 <\frac{11}{2}$. Thus, $L_{\mathcal{C}}(m)\leq5$.
\item For $m=8,9,\cdots,15$, $L_{\mathcal{C}}(m)\leq 4b+1 <7$. Thus, $L_{\mathcal{C}}(m)\leq6$.
\end{enumerate}
Thus, we have
\begin{equation}\label{eq23}
\begin{aligned}
  \sum_{m=1}^{\infty} \frac{1}{2^{L_{\mathcal{C}}(m)}}
              & = \sum_{m=1}^{16}\frac{1}{2^{L_{\mathcal{C}}(m)}} + \sum_{m=17}^{\infty}\frac{1}{2^{L_{\mathcal{C}}(m)}}  \\
              &\geq \frac{1}{2}+2\times\frac{1}{2^{3}}+4\times\frac{1}{2^{5}}+8\times\frac{1}{2^{6}}+\sum_{m=17}^{\infty}\frac{1}{2^{L_{\mathcal{C}}(m)}}  \\
              & =1+\sum_{m=17}^{\infty}\frac{1}{2^{L_{\mathcal{C}}(m)}}  \\
              & > 1.             \\
\end{aligned}
\end{equation}
This contradicts the Kraft's inequality~\cite{kraft}
\begin{equation}\label{eq12}
  \sum_{m=1}^{\infty} \frac{1}{2^{L_{\mathcal{C}}(m)}}\leq 1,
\end{equation}
so there is no such prefix code $\mathcal{C}$.
\end{proof}
\section{$\iota$ code to achieve $K_{\mathcal{C}}=2.5$}\label{sec_lb}
In this section, we provide a new UCI, termed $\iota$ code, to achieve $K_{\mathcal{C}}=2.5$.
First, we introduce some necessary notations.
Let $\alpha(m)$ be $m$ bits zeros followed by a single one, for all $m\in\mathcal{N}$.
Let $\beta(m)$ be the binary representation of $m\in\mathcal{N}$.
Let $[\beta(m)]$ be the binary string that removes the most significant bit one of $\beta(m)$.
For example, $\alpha(3)=0001$, $\beta(9)=1001$ and $[\beta(9)]=001$.
Let $\{0,1\}^{*}$ be a set containing all finite binary strings.

Next, the following defines an auxiliary code $\widetilde{\alpha}: \mathcal{N}\rightarrow \{0,1\}^{*}$.
\begin{equation}
\widetilde{\alpha}(m)=\left\{\begin{array}{lll}
1,                                         &\text{if } m=1,   \\
\alpha(\frac{m}{2})0,  &\text{if } m\geq 2 \text{ and }m \text{ is even,} \\
\alpha(\frac{m-1}{2})1,                   &\text{otherwise,}   \\
\end{array}\right.
\end{equation}
for all $m\in \mathcal{N}$.
Further, we define $\iota: \mathcal{N}\rightarrow \{0,1\}^{*}$ below.
\begin{equation}
\iota(m)=\widetilde{\alpha}(|\beta(m)|)[\beta(m)],
\end{equation}
for all $m\in \mathcal{N}$.
To better understand both codes, Table \ref{tab1} lists their first $16$ codewords.
\begin{table}[t]
\centering
\caption{The first $16$ codewords of $\widetilde{\alpha}$ code and $\iota$ code}\label{tab1}
\begin{tabular}{|c|c|c|}
\hline
$n$  &   $\widetilde{\alpha}$ code  & $\iota$ code   \\
\hline
$1$   &    1        &  1           \\
$2$   &    01 0      &  010 0      \\
$3$   &    01 1      &  010 1      \\
$4$   &    001 0    &  011 00     \\
$5$   &    001 1    &  011 01   \\
$6$   &    0001 0   &  011 10     \\
$7$   &    0001 1   &  011 11       \\
$8$   &    00001 0  &  0010 000    \\
$9$   &    00001 1  &  0010 001    \\
$10$  &    000001 0  &  0010 010    \\
$11$   &   000001 1  &  0010 011    \\
$12$   &   0000001 0  &  0010 100    \\
$13$   &   0000001 1  &  0010 101    \\
$14$   &   00000001 0  &  0010 110    \\
$15$   &   00000001 1  &  0010 111     \\
$16$   &   000000001 0 &  0011 0000   \\
\hline
\end{tabular}
\end{table}
From the definition, one can see that both code are prefix codes,
and the decoding algorithm naturally corresponds.

Then, we analyze the $K_{\iota}^{*}$ of $\iota$ code.
We obtain $L_{\iota}(1)=1$ and
\begin{equation}
\begin{aligned}
L_{\iota}(m)&=|\widetilde{\alpha}(1+\lfloor\log_{2}m\rfloor)|+\lfloor\log_{2}m\rfloor  \\
           &=2+\lfloor \frac{1+\lfloor\log_{2}m\rfloor}{2} \rfloor+\lfloor\log_{2}m\rfloor         \\
           &\leq \frac{3}{2}\lfloor\log_{2}m\rfloor+\frac{5}{2},
\end{aligned}
\end{equation}
for all $2\leq m\in \mathcal{N}$. Thus, we know that $\iota$ code is a UCI and $K_{\iota}^{*}\leq 2.5$ due to Theorem \ref{thm4}.
We consider the probability distribution $\overline{P}=(\frac{1}{2},\frac{1}{2})$, and we obtain
\begin{equation}
\frac{E_{\overline{P}}(L_{\iota})}{\max\{1,H(\overline{P})\}}=2.5.
\end{equation}
Thus, $K_{\iota}^{*}\geq 2.5$. Further, we have $K_{\iota}^{*}=2.5$.
We find the frist UCI such that $K_{\mathcal{C}}=2.5<2.75$.
This means that the range of $K_{\mathcal{C}}^{*}$ of the optimal UCI is improved to $2\leq K_{\mathcal{C}}^{*}\leq 2.5$.

Finally, we show that $\iota$ code is not asymptotically optimal.
We obtain $L_{\iota}(1)=1+\frac{3}{2}\lfloor\log_{2}1\rfloor$ and
\begin{equation}
\begin{aligned}
L_{\iota}(m)&=2+\lfloor \frac{1+\lfloor\log_{2}m\rfloor}{2} \rfloor+\lfloor\log_{2}m\rfloor         \\
           &> 1+\frac{3}{2}\lfloor\log_{2}m\rfloor,     \\
\end{aligned}
\end{equation}
for all $2\leq m\in \mathcal{N}$.
Due to Theorem \ref{thm3} and $L_{\iota}(m)\geq 1+\frac{3}{2}\lfloor\log_{2}m\rfloor$, for all $m\in \mathcal{N}$, $\iota$ code is not asymptotically optimal.
\section{A family of asymptotically optimal UCIs}\label{sec_new1}
In this section, we introduce a family of asymptotically optimal UCIs.
To better understand this family of asymptotically optimal UCIs, we first introduce a representative UCI in this family.
\subsection{$\kappa$ code to achieve $K_{\mathcal{C}}=\frac{8}{3}$}
In this subsection, we present an asymptotically optimal UCI, termed $\kappa$ code, to achieve $K_{\mathcal{\kappa}}=\frac{8}{3}<3.5$.
Notably, $\kappa$ code is a special case of a family of asymptotically optimal UCIs that will be introduced in the next subsection.

First, we define an auxiliary code $\widetilde{\gamma}: \mathcal{N}\rightarrow \{0,1\}^{*}$ below.
\begin{equation}
\widetilde{\gamma}(m)=\left\{\begin{array}{ll}
\widetilde{\alpha}(m),                     &\text{if } m<4,\\
\alpha(|\beta(m-2)|)[\beta(m-2)],                    &\text{otherwise,}   \\
\end{array}\right.
\end{equation}
for all $m\in \mathcal{N}$.
Further, we define $\kappa: \mathcal{N}\rightarrow \{0,1\}^{*}$ below.
\begin{equation}
\kappa(m)=\widetilde{\gamma}(|\beta(m)|)[\beta(m)],
\end{equation}
for all $m\in \mathcal{N}$.
Table \ref{tab2} lists some codewords for $\widetilde{\gamma}$ code and $\kappa$ code.
\begin{table}[t]
\centering
\caption{Some codewords of $\widetilde{\gamma}$ code and $\kappa$ code}\label{tab2}
\begin{tabular}{|c|c|c|}
\hline
$n$  &   $\widetilde{\gamma}$ code  & $\kappa$ code   \\
\hline
$1$   &    1        &  1           \\
$2$   &    01 0      &  010 0      \\
$3$   &    01 1      &  010 1      \\
$4$   &    001 0    &  011 00     \\
$5$   &    001 1    &  011 01   \\
$6$   &    0001 00   &  011 10     \\
$7$   &    0001 01   &  011 11       \\
$8$   &    0001 10  &  0010 000    \\
$9$   &    0001 11  &  0010 001    \\
$10$  &    00001 000  &  0010 010    \\
$11$   &   00001 001  &  0010 011    \\
$12$   &   00001 010  &  0010 100    \\
$20$   &   000001 0010  &  0011 0100    \\
$50$   &   0000001 10000    & 000100 10010    \\
$100$   &  00000001 100010  & 000101 100100     \\
\hline
\end{tabular}
\end{table}
From definitions, we know that $\widetilde{\gamma}$ code and $\kappa$ code are prefix codes, and the decoding algorithm naturally corresponds.
Due to the definition of $\widetilde{\gamma}$ code and $\kappa$ code,
we obtain
\begin{equation}
L_{\widetilde{\gamma}}(m)=\left\{\begin{array}{lll}
1,                                &\text{if } m=1,   \\
3,                                &\text{if } 2\leq m \leq 3,  \\
2+2\lfloor\log_{2}(m-2)\rfloor,   &\text{otherwise,}   \\
\end{array}\right.
\end{equation}
and
\begin{equation}
L_{\kappa}(m)=\left\{\begin{array}{llll}
1,                                &\text{if } m=1,   \\
4,                                &\text{if } 2\leq m \leq 3,  \\
5,                                &\text{if } 4\leq m \leq 7,  \\
2+\lfloor\log_{2}m\rfloor+2\lfloor\log_{2}(\lfloor\log_{2}m\rfloor-1)\rfloor,   &\text{otherwise.}   \\
\end{array}\right.
\end{equation}

Next, a lemma about the codeword length of $\kappa$ code is given.
\begin{lemma} \label{lemma3}
The codeword length of $\kappa$ code
\begin{equation}\label{eq51}
L_{\kappa}(m)\leq \frac{8}{3}+\frac{5}{3}\lfloor\log_{2}m\rfloor,
\end{equation}
for all $2\leq m\in\mathcal{N}$.
\end{lemma}
\begin{proof}
We first prove an auxiliary inequality as follows.
\begin{equation}\label{eq52}
\lfloor \log_{2}(x-1)\rfloor\leq\frac{1}{3}+\frac{1}{3}x,
\end{equation}
for all $3\leq x\in \mathcal{N}$. When $x=3$ or $x=4$, we can verify directly.
When $x=5$, both sides of inequality \eqref{eq52} are $2$.
Hereafter, if the left side of inequality \eqref{eq52} is increased by $1$, then $x$ must be increased by at least $4$.
At the same time, the right side of inequality \eqref{eq52} is increased by at least $\frac{1}{3}\times4=\frac{4}{3}>1$. Thus, inequality \eqref{eq52} holds.
For inequality \eqref{eq51}, when $m\leq7$, we can verify directly.
When $m\geq8$, we obtain
\begin{equation}
\begin{aligned}
L_{\kappa}(m)&=2+\lfloor\log_{2}m\rfloor+2\lfloor\log_{2}(\lfloor\log_{2}m\rfloor-1)\rfloor   \\
             &\leq 2+\lfloor\log_{2}m\rfloor+2\times(\frac{1}{3}+\frac{1}{3}\lfloor\log_{2}m\rfloor)    \\
             &= \frac{8}{3}+\frac{5}{3}\lfloor\log_{2}m\rfloor.  \\
\end{aligned}
\end{equation}
\end{proof}

Finally, we propose the main theorem in this subsection.
\begin{theorem}
\begin{itemize}
\item[(1)] $2.5\leq K_{\kappa}^{*}\leq \frac{8}{3}$;
\item[(2)] $\kappa$ code is asymptotically optimal.
\end{itemize}
\end{theorem}
\begin{proof}
\begin{itemize}
\item[(1)] Due to Theorem \ref{thm4} and Lemma \ref{lemma3}, we know that $\kappa$ code is a UCI and $K_{\kappa}^{*}\leq \frac{8}{3}$.
We consider $\overline{P}=(\frac{1}{2},\frac{1}{2})$, and we obtain
\begin{equation}
\frac{E_{\overline{P}}(L_{\kappa})}{\max\{1,H(\overline{P})\}}=2.5.
\end{equation}
Thus, $K_{\kappa}^{*}\geq 2.5$. Further, we have $2.5\leq K_{\kappa}^{*}\leq \frac{8}{3}$.
\item[(2)] The expected codeword length is
\begin{equation}~\label{eq55}
\begin{aligned}
  E_{P}(L_{\kappa})&= P(1)+4(P(2)+P(3))+5\sum_{n=4}^{7}P(n)+\sum_{n=8}^{\infty}P(n)L_{\kappa}(n)    \\
                   &< 5+\sum_{n=8}^{\infty}P(n)\log_{2}n+2\sum_{n=8}^{\infty}P(n)\log_{2}(\log_{2}n)      \\
                   & \leq 5+\sum_{n=2}^{\infty}P(n)\log_{2}n+2\sum_{n=2}^{\infty}P(n)\log_{2}(\log_{2}n)     \\
                   &\overset{(a)}{\leq} 5+H(P)+P(1)\log_{2}P(1)+2\sum_{n=2}^{\infty}P(n)\log_{2}(\log_{2}n)     \\
                   &\leq 5+H(P)+2P(1)\log_{2}1+2\sum_{n=2}^{\infty}P(n)\log_{2}(\log_{2}n)     \\
                   &\overset{(b)}{\leq} 5+H(P)+2\log_{2}\left(P(1)+\sum_{n=2}^{\infty}P(n)\log_{2}n\right)   \\
                   &\leq T_{\kappa}(H(P))\triangleq 5+H(P)+2\log_{2}(1+H(P)),  \\
\end{aligned}
\end{equation}
where $(a)$ is due to inequality \eqref{eq9} and $(b)$ is due to the convexity of the logarithm.
Therefore, we have
\begin{equation}
\lim\limits_{H(P)\to +\infty}R_{\kappa}(H(P))=\lim\limits_{H(P)\to +\infty}\frac{T_{\kappa}(H(P))}{H(P)}=1.
\end{equation}
And hence, $\kappa$ code is asymptotically optimal.
\end{itemize}
\end{proof}
\subsection{A family of asymptotically optimal UCIs}
In this subsection, we propose a family of asymptotically Optimal UCIs, termed $\kappa[t]$ code, to further reduce the upper bound of $K_{\mathcal{C}}^{*}$. First, we provide the relevant definition.
For any given positive integer $t$, we define a family of auxiliary codes $\widetilde{\gamma}[t]: \mathcal{N}\rightarrow \{0,1\}^{*}$ as follows:
\begin{equation}
\widetilde{\gamma}[t](m)=\left\{\begin{array}{ll}
\widetilde{\alpha}(m),                                         &\text{if } m<2t,\\
\alpha(|\beta(m+2-2t)|+t-2)[\beta(m+2-2t)],                    &\text{otherwise,}   \\
\end{array}\right.
\end{equation}
for all $m\in \mathcal{N}$.
Further, we define $\kappa[t]: \mathcal{N}\rightarrow \{0,1\}^{*}$ as follows:
\begin{equation}
\kappa[t](m)=\widetilde{\gamma}[t](|\beta(m)|)[\beta(m)],
\end{equation}
for all $m\in \mathcal{N}$. Two points need to be explained here.
One is the prefix of $\widetilde{\gamma}[t]$ code.
The codeword of $\widetilde{\gamma}[t]$ code starts with a series of consecutive zeros followed by a one.
From the definition of $\widetilde{\gamma}[t]$ code, we know that $\widetilde{\gamma}[t](2t-1)$ starts with $t-1$ consecutive zeros followed by a one,
and $\widetilde{\gamma}[t](2t)$ starts with $t$ consecutive zeros followed by a one. Thus, $\widetilde{\gamma}[t]$ code a prefix code.
The prefix of $\widetilde{\gamma}[t]$ code guarantees the prefix of $\kappa[t]$ code.
Their decoding algorithm naturally corresponds.
The other is the special case of these two familys of codes.
When $t=1$, $\widetilde{\gamma}[1]$ code is essentially Elias $\gamma$ code and $\kappa[1]$ code is essentially Elias $\delta$ code.
When $t=2$, $\widetilde{\gamma}[2]$ code is essentially $\widetilde{\gamma}$ code and $\kappa[2]$ code is essentially $\kappa$ code.

Due to the definition of $\widetilde{\gamma}[t]$ code and $\kappa[t]$ code,
we obtain
\begin{equation}
L_{\widetilde{\gamma}[t]}(m)=\left\{\begin{array}{lll}
1,                                &\text{if } m=1,   \\
2+\lfloor\frac{m}{2}\rfloor,                                &\text{if } 2\leq m < 2t,  \\
t+2\lfloor\log_{2}(m+2-2t)\rfloor,   &\text{otherwise,}   \\
\end{array}\right.
\end{equation}
and
\begin{equation}
L_{\kappa[t]}(m)=\left\{\begin{array}{lll}
1,                                &\text{if } m=1,   \\
2+\lfloor\log_{2}m\rfloor+\lfloor\frac{1+\lfloor\log_{2}m\rfloor}{2}\rfloor,                            &\text{if } 2\leq m < 2^{2t-1},  \\
t+\lfloor\log_{2}m\rfloor+2\lfloor\log_{2}(\lfloor\log_{2}m\rfloor+3-2t)\rfloor,   &\text{otherwise.}   \\
\end{array}\right.
\end{equation}

Next, a lemma about the codeword length of $\kappa[t]$ code is given.
\begin{lemma} \label{lemma4}
The codeword length of $\kappa[t]$ code
\begin{equation}\label{eq61}
L_{\kappa[t]}(m)\leq \frac{5}{2}+\frac{1}{2t+2}+\left(\frac{3}{2}+\frac{1}{2t+2}\right)\lfloor\log_{2}m\rfloor,
\end{equation}
for all $2\leq m\in\mathcal{N}$.
\end{lemma}
\begin{proof}
We first prove an auxiliary inequality as follows:
\begin{equation}\label{eq62}
t+2\lfloor \log_{2}(x+3-2t)\rfloor\leq\frac{5}{2}+\frac{1}{2t+2}+\left(\frac{1}{2}+\frac{1}{2t+2}\right)x,
\end{equation}
for all $2t-1\leq x\in \mathcal{N}$. When $x=2t-1$ or $x=2t$, we can verify directly.
When $x=2t+1$, both sides of inequality \eqref{eq62} are $4+t$.
Hereafter, if the left side of inequality \eqref{eq62} is increased by $2$, then $x$ must be increased by at least $4$.
At the same time, the right side of inequality \eqref{eq62} is increased by at least $\left(\frac{1}{2}+\frac{1}{2t+2}\right)\times4=2+\frac{2}{t+1}>2$. Thus, inequality \eqref{eq62} holds.
For inequality \eqref{eq61}, when $2\leq m<2^{2t-1}$, we have
\begin{equation}
\begin{aligned}
L_{\kappa[t]}(m)&=2+\lfloor\log_{2}m\rfloor+\lfloor\frac{1+\lfloor\log_{2}m\rfloor}{2}\rfloor   \\
                &\leq \frac{5}{2}+\frac{3}{2}\lfloor\log_{2}m\rfloor  \\
                & < \frac{5}{2}+\frac{1}{2t+2}+\left(\frac{3}{2}+\frac{1}{2t+2}\right)\lfloor\log_{2}m\rfloor. \\
\end{aligned}
\end{equation}
When $m\geq2^{2t-1}$, we obtain
\begin{equation}
\begin{aligned}
L_{\kappa[t]}(m)&=t+2\lfloor\log_{2}(\lfloor\log_{2}m\rfloor+3-2t)\rfloor+\lfloor\log_{2}m\rfloor  \\
             &\leq\frac{5}{2}+\frac{1}{2t+2}+\left(\frac{1}{2}+\frac{1}{2t+2}\right)\lfloor\log_{2}m\rfloor+\lfloor\log_{2}m\rfloor    \\
             &= \frac{5}{2}+\frac{1}{2t+2}+\left(\frac{3}{2}+\frac{1}{2t+2}\right)\lfloor\log_{2}m\rfloor.  \\
\end{aligned}
\end{equation}
\end{proof}

Finally, we propose the main theorem in this subsection.
\begin{theorem}
\begin{itemize}
\item[(1)] $2.5\leq K_{\kappa[t]}^{*}\leq 2.5+\frac{1}{2t+2}$;
\item[(2)] $\kappa[t]$ code is a family of asymptotically optimal UCIs.
\end{itemize}
\end{theorem}
\begin{proof}
\begin{itemize}
\item[(1)] Due to Theorem \ref{thm4} and Lemma \ref{lemma4}, we know that $\kappa[t]$ code is a UCI and $K_{\kappa[t]}^{*}\leq \frac{5}{2}+\frac{1}{2t+2}$.
We consider $\overline{P}=(\frac{1}{2},\frac{1}{2})$, and we obtain
\begin{equation}
\frac{E_{\overline{P}}(L_{\kappa[t]})}{\max\{1,H(\overline{P})\}}=2.5.
\end{equation}
Thus, $K_{\kappa[t]}^{*}\geq 2.5$. Further, we have $2.5\leq K_{\kappa[t]}^{*}\leq 2.5+\frac{1}{2t+2}$.

\item[(2)] When $t=1$, Elias~\cite{Elias75} has proven it.
When $t\geq2$, we obtain the following inequality derivation similar to \eqref{eq55}.
\begin{equation}
\begin{aligned}
  E_{P}(L_{\kappa[t]})&= \sum_{n=1}^{\infty}P(n)L_{\kappa[t]}(n)    \\
                   &< L_{\kappa[t]}(2^{2t-1}-1) +\sum_{n=2^{2t-1}}^{\infty}P(n)\log_{2}n+2\sum_{n=2^{2t-1}}^{\infty}P(n)\log_{2}(\log_{2}n)      \\
                   & \leq 3t-1+\sum_{n=2}^{\infty}P(n)\log_{2}n+2\sum_{n=2}^{\infty}P(n)\log_{2}(\log_{2}n)     \\
                   &\overset{(a)}{\leq} T_{\kappa[t]}(H(P))\triangleq 3t-1+H(P)+2\log_{2}(1+H(P)),  \\
\end{aligned}
\end{equation}
where $(a)$ is due to \eqref{eq55}.
Therefore, we have
\begin{equation}
\lim\limits_{H(P)\to +\infty}R_{\kappa[t]}(H(P))=\lim\limits_{H(P)\to +\infty}\frac{T_{\kappa[t]}(H(P))}{H(P)}=1.
\end{equation}
Thus, $\kappa[t]$ code is a family of asymptotically optimal UCIs.
\end{itemize}
\end{proof}
When $t$ tends to infinity, the value of $K_{\kappa[t]}=\frac{5}{2}+\frac{1}{2t+2}$ can be infinitely close to $2.5$.
An interesting thing needs to be explained here.
When $t$ is no longer a fixed value and tends to infinity, we can essentially regard $\lim\limits_{t\to +\infty}\kappa[t]$ code as $\iota$ code.
But at this time, $\lim\limits_{t\to +\infty}\kappa[t]$ code is not asymptotically optimal.

\section{$K_{\mathcal{C}}^{*}$ of the Classic UCIs}\label{sec_new}
In this section, we provide a more precise range of $K_{\mathcal{C}}^{*}$ of the classic UCIs by Theorem \ref{thm4}.
The main results of this section are summarized as follows.
\begin{theorem} \label{thm6}
\begin{itemize}
\item[(1)] $\delta$ code is asymptotically optimal UCI and $2.5 \leq K_{\delta}^{*}\leq 2.75 $;
\item[(2)] $\omega$ code is asymptotically optimal UCI and $2.1 < K_{\omega}^{*}\leq3 $;
\item[(3)] $\eta$ code is UCI and $2.5 \leq K_{\eta}^{*}\leq \frac{8}{3}$;
\item[(4)] $\theta$ code is asymptotically optimal UCI and $2.5 \leq K_{\theta}^{*}\leq 2.8$.
\end{itemize}
\end{theorem}
From Table \ref{tab5}, we only need to prove that $K_{\delta}^{*}\leq 2.75 $, $K_{\omega}^{*}\leq3 $, $K_{\eta}^{*}\leq \frac{8}{3}$ and $K_{\theta}^{*}\leq 2.8$.
We first prove the following lemma.
\begin{lemma} \label{lemma2}
For all $2\leq m \in\mathcal{N}$, we obtain
\begin{itemize}
\item[(1)] $L_{\delta}(m)\leq 2.75 + 1.75\lfloor\log_{2}m\rfloor$;
\item[(2)] $L_{\omega}(m)\leq 3 + 2\lfloor\log_{2}m\rfloor$;
\item[(3)] $L_{\eta}(m)\leq \frac{8}{3} + \frac{5}{3}\lfloor\log_{2}m\rfloor$;
\item[(4)] $L_{\theta}(m)\leq 2.8 + 1.8\lfloor\log_{2}m\rfloor$.
\end{itemize}
\end{lemma}
\begin{proof}
\begin{itemize}
\item[(1)] We prove the following inequality
\begin{equation}\label{eq30}
\lfloor \log_{2}(1+x )\rfloor \leq 0.875+0.375x,
\end{equation}
for all $x\in \mathcal{N}$. When $x\leq2$, we can verify directly.
When $x=3$, both sides of inequality \eqref{eq30} are $2$.
Hereafter, if the left side of inequality \eqref{eq30} is increased by $1$, then $x$ must be increased by at least $4$.
At the same time, the right side of inequality \eqref{eq30} is increased by at least $0.375\times4=1.5>1$. Thus, inequality \eqref{eq30} holds.
Further, we obtain
\begin{equation}
\begin{aligned}
L_{\delta}(m)&=1 + \lfloor \log_{2}m\rfloor + 2\lfloor \log_{2}(1+\lfloor\log_{2}m\rfloor)\rfloor         \\
             &\overset{(a)}{\leq} 1 + \lfloor \log_{2}m\rfloor + 2(0.875+0.375\lfloor \log_{2}m\rfloor )  \\
             &= 2.75 + 1.75\lfloor\log_{2}m\rfloor,   \\
\end{aligned}
\end{equation}
for all $2\leq m \in\mathcal{N}$, where $(a)$ is due to inequality \eqref{eq30}.
\item[(2)]  Our objective is to prove that
\begin{equation}
L_{\omega}(m)=1+\sum_{n=1}^{s}(\lambda^{n}(m)+1)\leq 3 + 2\lfloor\log_{2}m\rfloor,
\end{equation}
for all $2\leq m \in\mathcal{N}$. Let $a_{1}\triangleq2$ and $a_{m+1}\triangleq2^{a_{m}}$ for all $m \in\mathcal{N}$.
When $s\leq2$; that is, $a_{1}=2\leq m <16=a_{3}$, we can verify directly.
When $s=3$; that is, $a_{3}=16\leq m <65536=a_{4}$, since
\begin{equation}\label{eq33}
\lfloor \log_{2}x\rfloor \leq \frac{1}{2}x,
\end{equation}
for all $x \in\mathcal{N}$ and
\begin{equation}
\lfloor\log_{2}\lfloor \log_{2}x\rfloor\rfloor+1 \leq \frac{1}{2}x,
\end{equation}
for all $2\leq x \in\mathcal{N}$, we obtain
\begin{equation}
\begin{aligned}
L_{\omega}(m)&=3 + \lfloor \log_{2}m\rfloor + \lfloor \log_{2}\lfloor\log_{2}m\rfloor\rfloor+ (\lfloor \log_{2}\lfloor \log_{2}\lfloor \log_{2}m\rfloor\rfloor\rfloor+1)     \\
             &\leq 3 + \lfloor \log_{2}m\rfloor + \frac{1}{2}\lfloor \log_{2}m\rfloor+\frac{1}{2}\lfloor \log_{2}m\rfloor   \\
             &= 3 + 2\lfloor\log_{2}m\rfloor,
\end{aligned}
\end{equation}
for all $a_{3}\leq m <a_{4}$.
When $s\geq4$; that is, $m\geq a_{4}$, we consider the following three inequalities.
\begin{enumerate}
\item[2.1)] We have
\begin{equation}\label{eq36}
\begin{aligned}
\lambda^{2}(m)+1&=\lfloor \log_{2}\lfloor\log_{2}m\rfloor\rfloor+1    \\
                &\overset{(a)}{\leq}\frac{1}{2}\lfloor\log_{2}m\rfloor+1           \\
\end{aligned}
\end{equation}
for all $2\leq m\in\mathcal{N}$, where $(a)$ is due to inequality \eqref{eq33}.
\item[2.2)] We prove the following inequality
\begin{equation}\label{eq37}
\lambda^{3}(m)+1 \leq \frac{1}{4}\lfloor\log_{2}m\rfloor,
\end{equation}
for all $a_{4}\leq m\in \mathcal{N}$. When $m=a_{4}$, we obtain
\begin{equation}
3=\lambda^{3}(a_{4})+1<\frac{1}{4}\lfloor\log_{2}a_{4}\rfloor=4.
\end{equation}
Hereafter, if the left side of inequality \eqref{eq37} is increased by $1$, then $m$ must be increased by at least $2^{2^{2^{3}}}-2^{2^{2^{2}}}=2^{256}-2^{16}$.
At the same time, the right side of inequality \eqref{eq37} is increased by at least $\frac{1}{4}(2^{256}-2^{16})>1$. Thus, inequality \eqref{eq37} holds.
\item[2.3)] We prove the following inequality
\begin{equation}\label{eq39}
\lambda^{t}(m)+1 \leq \frac{1}{2^{t-1}}\lfloor\log_{2}m\rfloor,
\end{equation}
for all $a_{t}\leq m\in \mathcal{N}$, where $t$ is any given integer greater than or equal to $4$. When $m=a_{t}$, due to
\begin{equation}
\lambda^{t}(a_{t})=\lambda^{t-1}(a_{t-1})=\cdots=\lambda(a_{1})=1,
\end{equation}
we obtain
\begin{equation}
2=\lambda^{t}(a_{t})+1=\frac{1}{2^{3}}a_{3}\leq\frac{1}{2^{t-1}}a_{m-1}=\frac{1}{2^{t-1}}\lfloor\log_{2}a_{t}\rfloor.
\end{equation}
Hereafter, if the left side of inequality \eqref{eq39} is increased by $1$, then $m$ must be increased by at least $a_{t+1}-a_{t}$.
At the same time, the right side of inequality \eqref{eq39} is increased by at least
\begin{equation}
\begin{aligned}
&\frac{1}{2^{t-1}}(\lfloor\log_{2}a_{t+1}\rfloor-\lfloor\log_{2}a_{t}\rfloor)\\
=&\frac{1}{2^{t-1}}(a_{t}-a_{t-1})\geq \frac{1}{2^{3}}(a_{4}-a_{3})>1.
\end{aligned}
\end{equation}
Thus, inequality \eqref{eq39} holds.
\end{enumerate}
Due to inequality \eqref{eq36}, \eqref{eq37} and \eqref{eq39}, we obtain
\begin{equation}
\begin{aligned}
L_{\omega}(m)&=2+\lfloor\log_{2}m\rfloor+\sum_{n=2}^{s}(\lambda^{n}(m)+1)    \\
             &\leq 2+\lfloor\log_{2}m\rfloor+1+\sum_{n=1}^{s-1}\frac{\lfloor\log_{2}m\rfloor}{2^{n}}     \\
             &=3+(2-\frac{1}{2^{s-1}})\lfloor\log_{2}m\rfloor   \\
             &<3+2\lfloor\log_{2}m\rfloor.   \\
\end{aligned}
\end{equation}
\item[(3)]
When $m\leq3$, we can verify directly.
When $m\geq4$, due to $\lfloor\log_{2}m\rfloor\geq2$, we have
\begin{equation}
\begin{aligned}
L_{\eta}(m)&=3+\lfloor\log_2(m-1)\rfloor+\lfloor \frac{\lfloor\log_2(m-1)\rfloor}{2} \rfloor    \\
             &\leq \frac{8}{3}+\frac{1}{6}\times2+\frac{3}{2}\lfloor\log_{2}m\rfloor    \\
             &\leq \frac{8}{3}+\frac{1}{6}\lfloor\log_{2}m\rfloor +\frac{3}{2}\lfloor\log_{2}m\rfloor  \\
             &=\frac{8}{3} + \frac{5}{3}\lfloor\log_{2}m\rfloor.   \\
\end{aligned}
\end{equation}
\item[(4)]
We prove the following inequality
\begin{equation}\label{eq45}
0.2+1.5\lfloor \log_{2}x\rfloor \leq 0.8x,
\end{equation}
for all $3\leq x\in \mathcal{N}$. When $x=3$, we can verify directly.
When $x=4$, both sides of inequality \eqref{eq45} are $3.2$.
Hereafter, if the left side of inequality \eqref{eq45} is increased by $1.5$, then $x$ must be increased by at least $4$.
At the same time, the right side of inequality \eqref{eq45} is increased by at least $0.8\times4=3.2>1.5$. Thus, inequality \eqref{eq45} holds.
For $L_{\theta}(m)\leq 2.8 + 1.8\lfloor\log_{2}m\rfloor$, when $m\leq7$, we can verify directly.
When $m\geq8$, we obtain
\begin{equation}
\begin{aligned}
L_{\theta}(m)&=3+\lfloor \log_{2}m\rfloor+\lfloor\log_{2}\lfloor \log_{2}m\rfloor \rfloor+\lfloor\frac{\lfloor\log_{2}\lfloor \log_{2}m\rfloor \rfloor}{2}\rfloor   \\
             &\leq 3+\lfloor \log_{2}m\rfloor+1.5\lfloor\log_{2}\lfloor \log_{2}m\rfloor \rfloor     \\
             &=2.8+\lfloor \log_{2}m\rfloor+(0.2+1.5\lfloor\log_{2}\lfloor \log_{2}m\rfloor \rfloor)  \\
             &\leq 2.8 + 1.8\lfloor\log_{2}m\rfloor.   \\
\end{aligned}
\end{equation}
\end{itemize}
\end{proof}

Due to Lemma \ref{lemma2} and Theorem \ref{thm4}, we have $K_{\delta}^{*}\leq 2.75 $, $K_{\omega}^{*}\leq3 $, $K_{\eta}^{*}\leq \frac{8}{3}$ and $K_{\theta}^{*}\leq 2.8$.
Furthermore, Theorem \ref{thm6} is proved.
\begin{table}[t]
\centering
\caption{The latest research results for $K_{\mathcal{C}}^{*}$ of some UCIs}\label{tab3}
\begin{tabular}{|c|c|c|}
\hline
Code  & The range of $K_{\mathcal{C}}^{*}$ &  Asymptotically optimal   \\
\hline
$\gamma$ code  &     $K_{\gamma}^{*}=3$               &  No   \\
$\eta$ code    &   $ 2.5 \leq K_{\eta}^{*}\leq \frac{8}{3} $ &  No   \\
$\iota$ code  &     $K_{\iota}^{*}=2.5$               &  No   \\
$\delta$ code  &    $2.5 \leq K_{\delta}^{*}\leq 2.75 $  &  Yes  \\
$\omega$ code  &    $2.1 < K_{\omega}^{*}\leq3 $    &  Yes  \\
$\theta$ code   &   $ 2.5 \leq K_{\theta}^{*}\leq 2.8$ &  Yes  \\
$\kappa$ code   &   $ 2.5 \leq K_{\kappa}^{*}\leq \frac{8}{3}$ &  Yes  \\
$\kappa[t]$ code   &   $ 2.5 \leq K_{\kappa[t]}^{*}\leq 2.5+\frac{1}{2t+2}$ &  Yes  \\
\hline
\end{tabular}
\end{table}

From Theorem \ref{thm6}, $K_{\iota}^{*}=2.5$ and $2.5\leq K_{\kappa[t]}^{*}\leq 2.5+\frac{1}{2t+2}$,
we obtain Table \ref{tab3} to compare the expansion factor between our $\iota$ code, $\kappa[t]$ code and the classic UCIs previously proposed.
Currently, only $\iota$ code can achieve $K_{\mathcal{\iota}}=2.5$.
For asymptotically optimal UCIs, the current best result is that $\kappa[t]$ code can achieve $K_{\kappa[t]}=2.5+\frac{1}{2t+2}$, for all $t\in\mathcal{N}$.

\section{Conclusions}\label{sec_con}
In this paper, we study the expansion factor of UCI further, and Table \ref{tab3} summarizes the work of this paper.
From Table \ref{tab3}, the proposed $\iota$ code improves the expansion factor of optimal UCI to $K_{\mathcal{C}}=2.5$, and the proposed $\kappa[t]$ code improves the expansion factor of asymptotically optimal UCIs to $K_{\mathcal{C}}\Rightarrow 2.5$.
This work further reduces the range of the expansion factor to $2\leq K_{\mathcal{C}}^{*}\leq 2.5$.
There are several unresolved issues, as listed below.
\begin{enumerate}
\item one can see that the explicit value of $K_{\mathcal{C}}^{*}$ of the optimal UCI is still unknown.
\item $\omega$ code is the only UCI whose lower bound of $K_{\mathcal{C}}^{*}$ is less than $2.5$. Can $\omega$ code achieve $K_{\omega}<2.5$?
\end{enumerate}
\bibliographystyle{IEEEtran}
\bibliography{IEEEabrv,refs}

\begin{thebibliography}{10}
\providecommand{\url}[1]{#1}
\csname url@samestyle\endcsname
\providecommand{\newblock}{\relax}
\providecommand{\bibinfo}[2]{#2}
\providecommand{\BIBentrySTDinterwordspacing}{\spaceskip=0pt\relax}
\providecommand{\BIBentryALTinterwordstretchfactor}{4}
\providecommand{\BIBentryALTinterwordspacing}{\spaceskip=\fontdimen2\font plus
\BIBentryALTinterwordstretchfactor\fontdimen3\font minus
  \fontdimen4\font\relax}
\providecommand{\BIBforeignlanguage}[2]{{%
\expandafter\ifx\csname l@#1\endcsname\relax
\typeout{** WARNING: IEEEtran.bst: No hyphenation pattern has been}%
\typeout{** loaded for the language `#1'. Using the pattern for}%
\typeout{** the default language instead.}%
\else
\language=\csname l@#1\endcsname
\fi
#2}}
\providecommand{\BIBdecl}{\relax}
\BIBdecl

\bibitem{AC79}
J.~{Rissanen} and G.~G. {Langdon}, ``Arithmetic coding,'' \emph{IBM Journal of
  Research and Development}, vol.~23, no.~2, pp. 149--162, Mar. 1979.

\bibitem{AC84}
G.~G. {Langdon}, ``An introduction to arithmetic coding,'' \emph{IBM Journal of
  Research and Development}, vol.~28, no.~2, pp. 135--149, Mar. 1984.

\bibitem{H52}
D.~A. {Huffman}, ``A method for the construction of minimum-redundancy codes,''
  \emph{Proceedings of the IRE}, vol.~40, no.~9, pp. 1098--1101, Sep. 1952.

\bibitem{USC}
L.~Davisson, ``Universal noiseless coding,'' \emph{{IEEE} Trans. Inf. Theory},
  vol.~19, no.~6, pp. 783--795, Nov. 1973.

\bibitem{LZ77}
J.~Ziv and A.~Lempel, ``A universal algorithm for sequential data
  compression,'' \emph{{IEEE} Trans. Inf. Theory}, vol.~23, no.~3, pp.
  337--343, May 1977.

\bibitem{LZ78}
------, ``Compression of individual sequences via variable-rate coding,''
  \emph{{IEEE} Trans. Inf. Theory}, vol.~24, no.~5, pp. 530--536, Sep. 1978.

\bibitem{LZW}
Welch, ``A technique for high-performance data compression,'' \emph{Computer},
  vol.~17, no.~6, pp. 8--19, Jun. 1984.

\bibitem{No94}
L.~Gyorfi, I.~Pali, and E.~Van~der Meulen, ``There is no universal source code
  for an infinite source alphabet,'' \emph{{IEEE} Trans. Inf. Theory}, vol.~40,
  no.~1, pp. 267--271, Jan. 1994.

\bibitem{BY76}
J.~L. Bentley and A.~C.-C. Yao, ``An almost optimal algorithm for unbounded
  searching,'' \emph{Inf. Process. Lett.}, vol.~5, no.~3, pp. 82--87, Aug.
  1976.

\bibitem{AHK97}
R.~Ahlswede, T.~S. Han, and K.~Kobayashi, ``Universal coding of integers and
  unbounded search trees,'' \emph{{IEEE} Trans. Inf. Theory}, vol.~43, no.~2,
  pp. 669--682, Mar. 1997.

\bibitem{06}
J.~Zobel and A.~Moffat, ``Inverted files for text search engines,'' \emph{ACM
  Comput. Surv.}, vol.~38, no.~2, pp. 1--56, Jul. 2006.

\bibitem{DCC21}
L.~Allison, A.~S. Konagurthu, and D.~F. Schmidt, ``On universal codes for
  integers: Wallace tree, {Elias} omega and beyond,'' in \emph{Proc. 2021 Data
  Compression Conference (DCC)}, Mar. 2021, pp. 313--322.

\bibitem{DNA10}
K.~Daily, P.~Rigor, S.~Christley, X.~Xie, and P.~Baldi, ``Data structures and
  compression algorithms for high-throughput sequencing technologies,''
  \emph{{BMC} Bioinform.}, vol.~11, p. 514, Oct. 2010.

\bibitem{DNA13}
J.~J. Selva and X.~Chen, ``{SRC}omp: Short read sequence compression using
  burstsort and {Elias} omega coding,'' \emph{PLOS ONE}, vol.~8, no.~12, pp.
  1--7, 12 Dec. 2013.

\bibitem{Elias75}
P.~Elias, ``Universal codeword sets and representations of the integers,''
  \emph{{IEEE} Trans. Inf. Theory}, vol.~21, no.~2, pp. 194--203, Mar. 1975.

\bibitem{C1990}
R.~M. Capocelli, ``Flag encodings related to the zeckendorf representation of
  integers,'' in \emph{Sequences, Combinatorics, Compression, Security, and
  Transmission}.\hskip 1em plus 0.5em minus 0.4em\relax New York, NY, USA:
  Springer-Verlag, 1990, pp. 449--466.

\bibitem{AS17}
B.~T. {\'{A}}vila and R.~M.~C. de~Souza, ``{Meta-Fibonacci} codes: Efficient
  universal coding of natural numbers,'' \emph{{IEEE} Trans. Inf. Theory},
  vol.~63, no.~4, pp. 2357--2375, Apr. 2017.

\bibitem{A1993}
T.~Amemiya and H.~Yamamoto, ``A new class of the universal representation for
  the positive integers,'' \emph{IEICE Transactions on Fundamentals of
  Electronics, Communications and Computer Sciences}, vol. E76A, no.~3, pp.
  447--452, Mar. 1993.

\bibitem{L68}
V.~I. Levenshtein, ``On the redundancy and delay of decodable coding of natural
  numbers (in {Russian}),'' \emph{Problems of Cybernetics}, vol.~20, pp.
  173--179, 1968.

\bibitem{ER78}
S.~Even and M.~Rodeh, ``Economical encoding of commas between strings,''
  \emph{Commun. ACM}, vol.~21, no.~4, pp. 315--317, Apr. 1978.

\bibitem{S80}
Q.~F. Stout, ``Improved prefix encodings of the natural numbers (corresp.),''
  \emph{{IEEE} Trans. Inf. Theory}, vol.~26, no.~5, pp. 607--609, Sep. 1980.

\bibitem{Y00}
H.~Yamamoto, ``A new recursive universal code of the positive integers,''
  \emph{{IEEE} Trans. Inf. Theory}, vol.~46, no.~2, pp. 717--723, Mar. 2000.

\bibitem{yan}
W.~Yan and S.-J. Lin, ``On the minimum of the expansion factor for universal
  coding of integers,'' \emph{{IEEE} Trans. Commun.}, 2021,
  doi:10.1109/TCOMM.2021.3100497.

\bibitem{L81}
K.~B. Lakshmanan, ``On universal codeword sets,'' \emph{{IEEE} Trans. Inf.
  Theory}, vol.~27, no.~5, pp. 659--662, Sep. 1981.

\bibitem{AF87}
A.~Apostolico and A.~S. Fraenkel, ``Robust transmission of unbounded strings
  using {Fibonacci} representations,'' \emph{{IEEE} Trans. Inf. Theory},
  vol.~33, no.~2, pp. 238--245, Mar. 1987.

\bibitem{W88}
M.~Wang, ``Almost asymptotically optimal flag encoding of the integers,''
  \emph{{IEEE} Trans. Inf. Theory}, vol.~34, no.~2, pp. 324--326, Mar. 1988.

\bibitem{YO91}
H.~Yamamoto and H.~Ochi, ``A new asymptotically optimal code for the positive
  integers,'' \emph{{IEEE} Trans. Inf. Theory}, vol.~37, no.~5, pp. 1420--1429,
  Sep. 1991.

\bibitem{kraft}
L.~G. Kraft, ``A device for quantizing, grouping, and coding
  amplitude-modulated pulses,'' Master's thesis, Dept. of Electrical
  Engineering, Massachusetts Institute of Technology, Cambridge, Mass., 1949.

\end{thebibliography}

\end{document}